\documentclass{article}

\usepackage[OT2,T1]{fontenc}
\DeclareSymbolFont{cyrletters}{OT2}{wncyr}{m}{n}
\DeclareMathSymbol{\Sha}{\mathalpha}{cyrletters}{"58}
\usepackage{
	amsmath,
	amsthm,
	amssymb
}
\usepackage{graphicx}
\usepackage{booktabs} %\toprule etc.
\usepackage{tikz}
\usepackage{enumerate}
\usetikzlibrary{matrix}
\usepackage{lscape}
\usetikzlibrary{automata,positioning}
\usepackage{multicol}
\usepackage{float}
\usepackage{chngpage}
\usetikzlibrary{decorations.pathmorphing}
\tikzset{snake it/.style={decorate, decoration=snake}}
\usepackage{tikz-qtree}
\usepackage{mathtools}
\usepackage{tcolorbox}
\usepackage{collectbox}
\usepackage{pgf}
\usepackage{pgfplots}
\usepackage{longtable}
\usepackage{rotating}%sidewaystable
\usepackage{hyperref}

\usepackage{cleveref}

\usepackage{makecell}

\usepackage{pifont}% http://ctan.org/pkg/pifont
\newcommand{\eps}{\varepsilon}

\newcommand{\concat}{}

\usepackage[all]{xy}

\newtheorem{thm}{Theorem}
\newtheorem{df}[thm]{Definition}
\newtheorem{lem}[thm]{Lemma}
\newtheorem{cor}[thm]{Corollary}

\newtheorem{rem}[thm]{Remark}

\newcommand{\abs}[1]{\lvert#1\rvert}

\newcommand{\mt}{\mathtt}

\newcommand{\ul}{\underline}

\newcommand{\formalLanguage}{\mathsf}

\newcommand{\REP}{\formalLanguage{REP}}
\newcommand{\SQ}{\formalLanguage{SQ}}

\newcommand{\complexityClass}{\mathbf}

\newcommand{\NC}{\complexityClass{NC}}
\newcommand{\coCFL}{\complexityClass{coCFL}}
\newcommand{\coSAC}{\complexityClass{coSAC}}
\newcommand{\CFL}{\complexityClass{CFL}}

\newcommand{\DCFL}{\complexityClass{DCFL}}
\newcommand{\E}{\complexityClass{E}}

\newcommand{\coNP}{\complexityClass{coNP}}
\renewcommand{\P}{\complexityClass{P}}
\renewcommand{\E}{\complexityClass{E}}

\newcommand{\SAC}{\complexityClass{SAC}}
\newcommand{\OSAC}{{\oplus\complexityClass{SAC}}}

\begin{document}
\title{Maximal automatic complexity and context-free languages}
\author{Bj{\o}rn Kjos-Hanssen\thanks{This work was partially supported by a grant from the Simons Foundation (\#704836 to Bj\o rn Kjos-Hanssen).}}
\maketitle

\begin{abstract}
Let $A_N$ denote nondeterministic automatic complexity and
\[
	L_{k,c}=\{x\in [k]^* : A_N(x)> \abs{x}/c\}.
\]
In particular, $L_{k,2}$ is the language of all $k$-ary words for which $A_N$ is maximal, while $L_{k,3}$ gives a rough dividing line between complex and simple.
Let $\CFL$ denote the complexity class consisting of all context-free languages.
While it is not known that $L_{2,2}$ is infinite, Kjos-Hanssen (2017) showed that $L_{3,2}$ is $\CFL$-immune but not $\coCFL$-immune.
We complete the picture by showing that $L_{3,2}\not\in\coCFL$.

Turning to Boolean circuit complexity, we show that $L_{2,3}$ is $\SAC^0$-immune and $\SAC^0$-coimmune.
Here $\SAC^0$ denotes the complexity class consisting of all languages
computed by (non-uniform) constant-depth circuits with semi-unbounded fanin.

As for arithmetic circuits, we show that $\{x:A_N(x)>1\}\not\in\OSAC^0$.
In particular, $\SAC^0\not\subseteq\oplus \SAC^0$, which resolves an open implication from the Complexity Zoo.
\end{abstract}

	\section{Introduction}
		Automatic complexity is a computable form of Kolmogorov complexity which was introduced by Shallit and Wang in 2001~\cite{MR1897300}.
		It was studied by Jordon and Moser in 2021~\cite{fct}, and in a series of papers by the author and his coauthors.
		Roughly speaking, the automatic complexity $A(x)$ of a string is the minimal number of states of an automaton accepting only $x$ among its equal-length peers.
		The nondeterministic version $A_N$ was introduced by Hyde~\cite{MR3386523}:
		\begin{df}\label{def:NFA}
			Let $\Sigma$ be finite a set called the \emph{alphabet}\index{alphabet} and let $Q$ be a finite set whose elements are called \emph{states}.\index{states}
			A \emph{nondeterministic finite automaton} (NFA)\index{NFA} is a 5-tuple
			\(
				M=(Q,\Sigma,\delta,q_0,F).
			\)
			The \emph{transition function}\index{transition function} $\delta:Q\times\Sigma\to\mathcal P(Q)$ maps each $(q,b)\in Q\times\Sigma$ to a subset of $Q$.
			Within $Q$ we find the \emph{initial state} $q_0\in Q$ and
			the set of \emph{final states} $F\subseteq Q$.
			The function $\delta$ is extended to a function $\delta^*:Q\times\Sigma^*\to\mathcal P(Q)$
			by\footnote{We denote concatenation by %$\sigma\concat\tau$ or by
			juxtaposition, $\sigma\tau$.}
			\[
				\delta^*(q,\sigma\concat i)=\bigcup_{s\in \delta^*(q,\sigma)}\delta(s,i).
			\]
			Overloading notation we also write $\delta=\delta^*$.
			The \emph{language accepted by $M$} is
			\[
				L(M)=\{x \in \Sigma^*: \delta(q,x)\cap F\ne\emptyset\}.
			\]
			A \emph{deterministic finite automaton} (DFA)\index{DFA} is also a 5-tuple
			\(
				M=(Q,\Sigma,\delta,q_0,F).
			\)
			In this case, $\delta:Q\times\Sigma\to Q$ is a total function and is extended to $\delta^*$ by
			$\delta^*(q,\sigma\concat i)=\delta(\delta^*(q,\sigma),i)$.
			If the domain of $\delta$ is a subset of $Q\times\Sigma$, $M$ is a \emph{partial DFA}.
			Finally,
			the set of words accepted by $M$ is
			\[
				L(M)=\{x \in \Sigma^*: \delta(q,x)\in F\}.
			\]
		\end{df}
		\begin{df}[{\cite{MR3386523,MR1897300}}]\label{precise}
			Let $L(M)$ be the language recognized by the automaton $M$.
			The nondeterministic automatic complexity $A_N(w)$ of a word $w$ is the minimum number of states of an NFA $M$
			such that $M$ accepts $w$ and the number of paths along which $M$ accepts words of length $\abs{w}$ is 1.
		\end{df}
		The fundamental upper bound on $A_N$ was discovered by Kayleigh Hyde in 2013:
		\begin{thm}[{Hyde~\cite{MR3386523}}]\label{Hyde}
			Let $x$ be any word of length $n\in\mathbb N$.
			Then
			\[
				A_N(x) \le {\lfloor} n/2 {\rfloor} + 1\text{.}
			\]
		\end{thm}
		The sharpness of \Cref{Hyde} over a ternary alphabet is also known~\cite{MR3386523}.
		Therefore, if $A_N(x) < {\lfloor} n/2 {\rfloor} + 1$, we say that $x$ is \emph{$A_N$-simple} (see \Cref{april21-2020}).

		The question whether the computation of automatic complexity is itself computationally tractable was raised by Shallit and Wang and also by Allender~\cite{MR3629714}.
		Here we show that a natural decision problem associated with automatic complexity does not belong to the complexity class
		$\SAC^0$ consisting of languages accepted by constant-depth Boolean circuits with semi-unbounded fan-in.
		We also confirm~\cite[Conjecture 11]{MR3712310} by showing that the set of maximally complex words over a ternary alphabet is not co-context free.

	\section{Context-free languages}

		A word is called \emph{squarefree} if it cannot be written as $xyyz$ where $y$ is nonempty.
		In ruling out context-freeness of the set of words of nonmaximal nondeterministic automatic complexity, our method of proof goes back to the 1980s.
		It consists in analyzing the proof in~\cite{MR784746}
		in the form given in Shallit 2008~\cite{Shallit:2008:SCF:1434864}, to conclude that it pertains to any language $L$ with $\SQ_3\subseteq L\subseteq\REP_3$.
		Here $\SQ_k$ is the set of square words over the alphabet $[k]$ and $\REP_k$ is the set of words over $[k]$ that are repetitive, i.e., not squarefree.

		Our key tool will be the Interchange \Cref{lem:inter}, originally due to Ehrenfeucht and Rozenberg~\cite{MR701985}; see also Berstel and Boasson~\cite{MR1127187}.
		We denote the length of a word $x$ by $\abs{x}$, and let $[k]=\{0,1,\dots,k-1\}$ for a nonnegative integer $k$.

		\begin{lem}[Interchange lemma~\cite{MR784746}]\label{lem:inter}
			For each $L\in\CFL$ there is an integer $c>0$ such that for all integers $n\ge 2$, all subsets $R\subseteq L\cap\Sigma^n$,
			and all integers with $2\le m\le n$, there exists a subset $Z\subseteq R$, $Z=\{z_1,z_2,\dots,z_k\}$ such that
			$k\ge\frac{\abs{R}}{c(n+1)^2}$ and such that there exist decompositions $z_i=w_i x_i y_i$, $1\le i\le k$, such that for all $1\le i,j\le k$,
			\begin{enumerate}[(a)]
				\item $\abs{w_i}=\abs{w_j}$,
				\item $\abs{y_i}=\abs{y_j}$,
				\item $\frac{m}2<\abs{x_i}=\abs{x_j}\le m$; and
				\item $w_i x_j y_i\in L$.
			\end{enumerate}
		\end{lem}

		\begin{rem}
			When establishing \Cref{77} and \Cref{april21-2020} we needed a word with a certain property.
			We first found an example over a 6-letter alphabet, $(\mt{123})^2\mt{0}(\mt{12345})^2$.
			We were able to reduce the alphabet size from 6 to 5 by using the slightly modified word $(\mt{123})^2\mt{0}(\mt{12341})^2$.
			We expect that the alphabet size, and word length, can be reduced further, but we do not pursue it here.
		\end{rem}
		\begin{lem}\label{77}
			Let
			\(
				w=(\mt{123})^2\mt{0}(\mt{12341})^2\in \{\mt 0,\mt 1,\mt 2,\mt 3,\mt 4\}^{17}.
			\)
			Then $A_N(w)=8$.
		\end{lem}
		\begin{proof}
			Let $M$ be the following NFA.
			
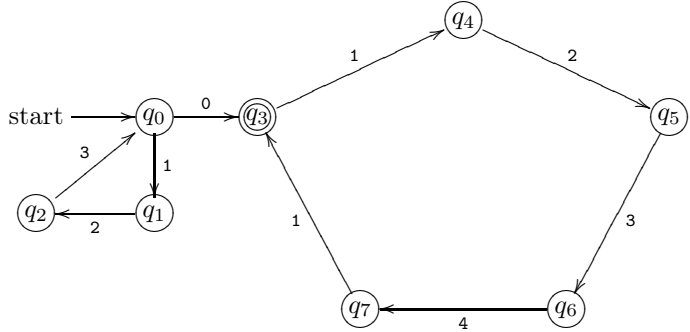
\[
			\xymatrix{
			&&& & *+[Fo]{q_4} \ar[drr]^{\mt 2} \\
			\text{start}\ar[r]&*+[Fo]{q_0}\ar[r]^{\mt 0}\ar[d]^{\mt 1}&*+[Foo]{q_3}\ar[urr]^{\mt 1}& & & &*+[Fo]{q_5} \ar[ddl]^{\mt 3} \\
			*+[Fo]{q_2}\ar[ur]^{\mt 3} & *+[Fo]{q_1}\ar[l]^{\mt 2}\\
			&&&*+[Fo]{q_7} \ar[uul]^{\mt 1}& & *+[Fo]{q_6}\ar[ll]^{\mt 4} \\
			}
			\]
			Since the equation
			\(
			3x+1+5y=17
			\)
			has the unique solution $(x,y)=(2,2)$ over $\mathbb N$, we see that $M$ accepts $w$ uniquely via the sequence of states
			\[
				(q_0, q_1, q_2, q_0, q_1, q_2, q_0, q_3, q_4, q_5, q_6, q_7, q_3, q_4, q_5, q_6, q_7, q_3).
			\]
			Since $M$ has only 8 states, it follows that $A_N(w)\le 8$.
			Conversely, the inequality $A_N(w)\ge 8$ was verified with a \emph{brute force} computerized search in {Python 3.7} on {July 7, 2021}.
		\end{proof}

		\begin{df}\label{df:power}
			For a word $x$ and positive integers $p,q$, we define the power $x^{\alpha}$, $\alpha=p/q$, to be the prefix of length $n$ of $x^{\omega} = xxx\cdots$,
			where $n = p\abs{x}/q$ is an integer, if it exists. The word $x^{\alpha}$ is called an $\alpha$-power; if $\alpha=2$ this is read as squarefree, if $\alpha=3$ cubefree.
			If a word $w$ has a (consecutive) subword of the form $x^{\alpha}$ then $x$ is said to occur with exponent $\alpha$ in $w$.
			The word $\mathbf{w}$ is $\alpha$-power-free if it contains no subwords which are  $\alpha$-powers.
			Finally, $w$ is \emph{overlap-free} if it is $\alpha$-power-free for all $\alpha>2$.
		\end{df}
		As an example of \Cref{df:power}, we have $(\mt{0110})^{3/2}=\mt{011001}$.

		\begin{lem}\label{april21-2020}
			There exists an $A_N$-simple overlap-free word of odd length.
		\end{lem}
		\begin{proof}
			Consider the word $w$ of \Cref{77}, of length 17.
		\end{proof}

		\begin{lem}\label{jun20-2022}
			Let $\Sigma$ be a finite alphabet.
			Let $\alpha\in\mathbb Q$, $\alpha\ge 1$ and let $x\in\Sigma^*$ such that $x$ is an $\alpha$-power.
			Then $A_N(x)\le\abs{x}/\alpha$.
		\end{lem}
		\begin{proof}
			Let $x=(x_1\cdots x_v)^{\alpha}$ where $m,v\in\mathbb N$ such that the fractional part of $\alpha$ is $m/v$ and $x_1,\dots,x_v\in\Sigma$.
			Let $M$ be the following NFA, whose digraph is a cycle:
			
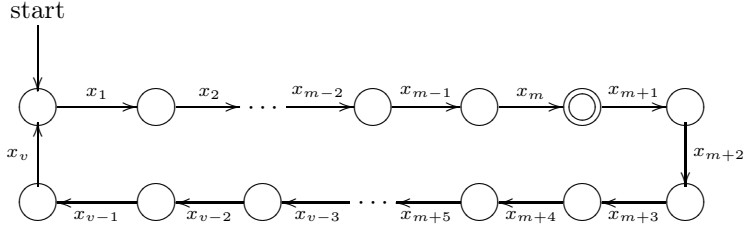
\[
				\xymatrix{
				\text{start}\ar[d]\\
				*+[Fo]{\phantom{q_0}}\ar[r]^{x_1}	&	*+[Fo]{\phantom{q_0}}\ar[r]^{x_2}	&	\cdots\ar[r]^{x_{m-2}}	&	*+[Fo]{\phantom{q_0}}\ar[r]^{x_{m-1}}	&	*+[Fo]{\phantom{q_0}}\ar[r]^{x_m}	&	*+[Foo]{\phantom{q_0}}\ar[r]^{x_{m+1}}	&	*+[Fo]{\phantom{q_0}}\ar[d]^{x_{m+2}}\\
				*+[Fo]{\phantom{q_0}}\ar[u]^{x_{v}}	&	*+[Fo]{\phantom{q_0}}\ar[l]^{x_{v-1}}	&	*+[Fo]{\phantom{q_0}}\ar[l]^{x_{v-2}}	&	\cdots\ar[l]^{x_{v-3}}			&	*+[Fo]{\phantom{q_0}}\ar[l]^{x_{m+5}}	&	*+[Fo]{\phantom{q_0}}\ar[l]^{x_{m+4}}	&	*+[Fo]{\phantom{q_0}}\ar[l]^{x_{m+3}}
				}
			\]
			The number of states is $v$, so $A_N(x)\le v=\abs{x}/\alpha$.
		\end{proof}
		\begin{lem}\label{feb27-2022-6pm}
			Let $\Sigma$ be a finite alphabet.
			Consider the following implication for a word $x\in\Sigma^*$.
			\begin{equation}\label{it}
				\text{
					$A_N(x)\le \abs{x}/\alpha \implies$ $x$ contains an $\alpha$-power.
				}
			\end{equation}
			\begin{enumerate}[(i)]
				\item Let $\alpha\ge 1$ be an integer. Then~\eqref{it} holds.
				\item \eqref{it} fails for $\alpha=2+\frac18$.
			\end{enumerate}
		\end{lem}
		\begin{proof}

			\noindent Proof of (i):
			Suppose that $A_N(x)\le\abs{x}/k$ as witnessed by an NFA $M$.
			Then the $\abs{x}+1$ visiting times of $M$ during its computation on input $x$ are distributed among the at most $\abs{x}/k$ states.
			Consequently, some state is visited at least $\frac{\abs{x}+1}{\abs{x}/k} = k(1+\frac1{\abs{x}})>k$, and hence at least $k+1$, times.
			However, \cite[Theorem 16]{MR3386523} says: If an NFA $M$ uniquely accepts $x$ of length $n$, and visits a state $p$ at least $k+1$ times, where $k\ge 2$,
			then $x$ contains a $k$th power.
			
			\noindent Proof of (ii): consider the word $x$ of \Cref{77} of length 17 and let $\alpha=2+\frac18$.
			Note that $\abs{x}/\alpha=17/(2+\frac18)=8$.
		\end{proof}

		Let $\Sigma,\Delta$ be finite alphabets. A \emph{morphism} $\varphi:\Sigma^*\to\Delta^*$
		is a function satisfying $\varphi(xy)=\varphi(x)\varphi(y)$ for all $x,y\in\Sigma^*$.
		Note that in order to define a particular morphism $\varphi$ it suffices to define $\varphi(a)$ for each $a\in\Sigma$.

		A function $f$ is called squarefree-preserving if for each squarefree $x$, $f(x)$ is squarefree.

		\begin{thm}[Brandenburg~\cite{MR693069}]\label{thm:brandenburg}
			The following morphism $h$ is squarefree-preserving:
			\begin{eqnarray*}
				\mt{0}\to	\mt{0102012021012102010212}\\
				\mt{1}\to	\mt{0102012021201210120212}\\
				\mt{2}\to	\mt{0102012101202101210212}\\
				\mt{3}\to	\mt{0102012101202120121012}\\
				\mt{4}\to	\mt{0102012102010210120212}\\
				\mt{5}\to	\mt{0102012102120210120212}
			\end{eqnarray*}
		\end{thm}
		
		\begin{df}
			The perfect shuffle of two words $x$, $y$ of the same length $a$ is $\Sha(x,y)=x_1 y_1 \cdots x_a y_a$.
		\end{df}

		We follow Shallit's presentation~\cite[Theorem 4.5.4]{Shallit:2008:SCF:1434864} in \Cref{what we already proved}, which is a slight strengthening thereof.
		\begin{thm}\label{what we already proved}
			For each language $L$,
			\[
				\SQ_6\subseteq L\subseteq \REP_6\implies L\not\in\CFL.
			\]
		\end{thm}
		\begin{proof}
			Assume that $L\in\CFL$. Let $c$ be the constant in the Interchange \Cref{lem:inter} and choose $n$ divisible by 8 and sufficiently large so that
			\[
			\frac{2^{n/4}}{c(n+1)^2} > 2^{n/8}.
			\]
			Let $r'$ be a squarefree word over the alphabet $\{\mt 0,\mt 1,\mt 2\}$ of length $\frac{n}4-1$ and let $r=\mt 3\,r'$,
			so that $\abs{r}=\frac{n}4$. Let 
			\[
				A_n = \{\Sha(rr, s): s\in\{\mt 4,\mt 5\}^{n/2}\}\subseteq\{\mt 0,\dots,\mt 5\}^n.
			\]
			As shown in~\cite{Shallit:2008:SCF:1434864},
			\begin{enumerate}
				\item if $z_i=w_i x_i y_i\in A_n$ for $i\in\{1,2\}$ with $\abs{w_1}=\abs{w_2}$, $\abs{x_1}=\abs{x_2}$, and $\abs{y_1}=\abs{y_2}$, then
					$w_1 x_2 y_1\in A_n$ and $w_2 x_1 y_2\in A_n$, too; and
					\item if $z\in A_n$, then $z$ contains a square if and only if $z$ is a square.
			\end{enumerate}
			Now let $B_n=A_n\cap SQ_6=A_n\cap \REP_6=\{\Sha(rr,ss):s\in\{\mt 4,\mt 5\}^{n/4}\}$ and note that $\abs{B_n}=2^{n/4}$.
			By the Interchange \Cref{lem:inter} with $m=n/2$ and $R=B_n$,
			there is a subset $Z\subseteq B_n$, $Z=\{z_1,\dots,z_k\}$ with $z_i=w_i x_i y_i$ satisfying the conclusions of that lemma.
			In particular,
			\[
				k=\abs{Z}\ge \frac{\abs{R}}{c(n+1)^2} = \frac{2^{n/4}}{c(n+1)^2}>2^{n/8}.
			\]

			Case 1: There exist indices $g$, $h$ such that $x_g\ne x_h$.
			By the Interchange \Cref{lem:inter}, $w_g x_h y_g\in L$. Moreover, one of the $\mt 4$ or $\mt 5$'s was changed going from $x_g$ to $x_h$.
			Since $\abs{x}\le m=n/2$, the corresponding $\mt 4$ or $\mt 5$ in the other half was not changed.
			So $w_g x_h y_g$ is not a square and hence does not contain a square. So $w_g x_h y_g\not\in L$, contradiction.

			Case 2: Case 1 fails.
			Then all the $x_i$ are the same and have length at least $m/2=n/4$.
			Therefore there are at least $n/4$ positions in which all the $z_i$ are the same.
			Because of the shuffle and since $n/4$ is even, at least half of these, i.e., at least $n/8$ positions, contain $\mt 4$'s and $\mt 5$'s,
			which means that these $n/8$ positions have constant values within $Z$.
			This leaves at most at most $n/8$ positions where both choices $\mt 4,\mt 5$ are available within $Z$.
			Thus $\abs{Z}\le 2^{n/8}$, which is a contradiction.
		\end{proof}

		\begin{thm}[{\cite[Theorem 4.1.4]{Shallit:2008:SCF:1434864}}]\label{414}
			Let $h$ be a morphism and $L\in\CFL$. Then $h^{-1}(L)=\{x\mid h(x)\in L\}\in\CFL$.
		\end{thm}

		\begin{thm}\label{mainish}
			For each language $L$,
			\[
				\SQ_3\subseteq L\subseteq \REP_3\implies L\not\in\CFL.
			\]
		\end{thm}
		\begin{proof}
			Assume that $\SQ_3\subseteq L\subseteq \REP_3$.
			Let $h$ be the squarefree-preserving morphism of \Cref{thm:brandenburg}.
			Since $h$ is a morphism, for all $x,w$ we have
			\[
				w=xx\implies h(w)=h(x)h(x)\in SQ_3.
			\]
			Thus $\SQ_6\subseteq h^{-1}(SQ_3)$.
			Also, the statement that $h$ is squarefree-preserving is equivalent to: $h^{-1}(\REP_3)\subseteq \REP_6$.
			Thus
			\[
				\SQ_6\subseteq h^{-1}(SQ_3) \subseteq h^{-1}(L)\subseteq h^{-1}(\REP_3) \subseteq \REP_6.
			\]
			By \Cref{what we already proved}, $h^{-1}(L)\not\in\CFL$. Hence by \Cref{414}, $L\not\in\CFL$.
		\end{proof}

		\begin{cor}
			$\{x\in\{\mt 0,\mt 1,\mt 2\}^*: A_N(x)\le\abs{x}/2\}\not\in\CFL$.
		\end{cor}
		\begin{proof}
			Let $L=\{x\in\{\mt 0,\mt 1,\mt 2\}^*: A_N(x)\le\abs{x}/2\}$.
			By \Cref{jun20-2022} and \Cref{feb27-2022-6pm} with $\alpha=2$,
			\(
				\SQ_3\subseteq L\subseteq \REP_3.
			\)
			By \Cref{mainish}, $L\not\in\CFL$.
		\end{proof}

	\section{Sensitivity of automatic complexity}
		A natural decision problem associated to automatic complexity can be defined as follows.
		\begin{df}\label{def:cb}
			Let $k,c\ge 1$ be integers, and let $A_N$ denote nondeterministic automatic complexity.
			We define the language $L_{k,c}$ by
			\[
				L_{k,c}=\{x\in [k]^* : A_N(x)> \abs{x}/c\}.
			\]
		\end{df}
		The intent is that the decision problem $L_{2,3}$ captures the ``random or not'' characteristic of a word $x$.
		Shallit and Wang~\cite{MR1897300} asked whether automatic complexity can be computed in polynomial time.
		This remains open, in particular so does the question whether $L_{2,3}\in \P$.
		It is easy to see that $L_{2,3}\in\coNP$~\cite{MR3386523}.
		Our state of knowledge is shown in \Cref{state-of-knowledge}.

		\begin{figure}
			\[
				\xymatrix@C=0.5em{
				&& \coNP\cap\E &\\
				&& \P\ar[u] &\\
				&L_{2,3}\ar@{-->}[uur]\ar@{..>}[dl]\ar@{..>}[d]& \SAC^1\ar[u]&L_{3,2}\ar@{-->}[uul]\ar@{..>}[dr]\ar@{..>}[d]\\ %=LOGCFL/poly
				\SAC^0\ar[urr]&\coSAC^0\ar[ur]&&\CFL\ar[ul] & \coCFL\ar[ull]\\
				\NC^0\ar[u]\ar[ur]&&
				&& \DCFL\ar[u]\ar[ul]
				}
			\]
			\caption{
			Known implications ($\subseteq$) denoted by solid arrow, containments ($\in$) by dashed arrow, non-containments ($\not\in$) by dotted arrow.
			The non-containments are demonstrated in this paper for the first time except for the one from $L_{3,2}$ to $\CFL$.
			}\label{state-of-knowledge}
		\end{figure}
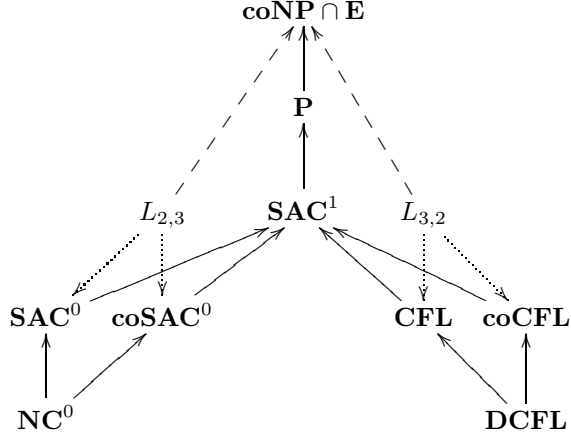
		In order to prove the remaining main result of this paper, \Cref{flipTurn}, we shall use \Cref{thm:july31} on the abundance of high-complexity words, and a new result
		\Cref{feb17-pm-2022}. The latter states that for any $x\in\{0,1\}^n$, and $a_1<\cdots<a_c<n$, there exists $y\in\{0,1\}^n$ which agrees with $x$ on the $a_i$,
		and has $A_N(y)=O(c^2\log n)$ where the constant $O$ is independent of $c$ and $n$.
		Before starting the proof, the reader is invited to consider \Cref{tab:jun20-2022}.
		\begin{table}
			\centering
			\begin{tabular}{c c c c c c c c}
				\toprule
				$c\setminus n$   &0&1&2&3&4&5&6\\
				\midrule
				0				 &1&1&1&1&1&1&1\\
				1				 &-&1&1&1&1&1&1\\
				2				 &-&-&2&2&2&3&3\\
				3				 &-&-&-&2&3&3&4\\
				4				 &-&-&-&-&3&3&4\\
				5				 &-&-&-&-&-&3&4\\
				6				 &-&-&-&-&-&-&4\\
			\bottomrule
			\end{tabular}
		\caption{The best bound on $A_N(y)$ as a function of $c$ and $n$ for some small values of $c$ and $n$.}\label{tab:jun20-2022}
		\end{table}

		We now prove some lemmas that will culminate in \Cref{feb17-pm-2022}.
		\Cref{lem:MA} provides a small but valuable refinement of what we would otherwise get more crudely in \Cref{feb16-2022}.
		\begin{lem}[Max Alekseyev~\cite{416457}]\label{lem:MA}
			Let $c\ge 2$ be an integer and let $a_1<a_2<\cdots<a_c$ be real numbers.
			The average of the values $a_j-a_i$, $i<j$, is at most $\frac{c}{2(c-1)}\ell$, where $\ell=a_c-a_1$.
		\end{lem}
		\begin{proof}
			The average of the values $a_j-a_i$, $i<j$ is
			\begin{eqnarray}
			\binom{c}2^{-1}\sum_{i<j}\abs{a_j-a_i}&=&\frac{2}{c^2-c}\sum_{i<j}a_j-a_i\nonumber \\
			&=&\frac{2}{c(c-1)}\sum_{i=0}^{(c-1)/2}(c-1-2i)(a_{c-i}-a_{i+1})\\
			&\le& \frac{2\ell}{c(c-1)}\sum_{i=0}^{(c-1)/2}(c-1-2i)\label{MA1} \\
			&\le& \frac{c}{2(c-1)}\ell \label{MA2}.
			\end{eqnarray}
			For~\eqref{MA1},
			the ``large'' values $a_c,a_{c-1},\dots,a_{c-i},\dots,a_{c-(c-1)/2}$ are the larger of the two numbers in a difference $c-1-i$ times
			and the smaller of the two $i$ times. The ``small'' values $a_{i+1}$ are the smaller $c-1-i$ times and the larger $i$ times.

			For~\eqref{MA2}, if $c\in\{2d,2d+1\}$, the sum is
			\begin{eqnarray*}
				\sum_{i=0}^{d-1} (c - (2i+1)) = cd - d^2 &=& \begin{cases} d^2,& c=2d\\ d^2+d,& c=2d+1\end{cases}
			\\
				\le (c/2)^2 &=& \begin{cases} d^2,& c=2d\\ (d+\frac12)^2,& c=2d+1.\end{cases}
			\end{eqnarray*}
		\end{proof}
		\begin{df}
			Let $p_n$ denote the $n$th prime number: $p_1=2, p_2=3$, etc.
			The \emph{primorial function} $x\#$ is defined by $x\# = \prod_{i=1}^q p_q$, where $q$ is maximal such that $p_q\le x$.
		\end{df}
		\begin{lem}\label{feb16-2022}
			Let $0\le a_1<\dots < a_c<n$, all integers.
			Let $q\in\mathbb N$ be such that $n\le \frac{2(c-1)}c (p_q\#)^{1/\binom{c}2}$.
			There exists a modulus $m$ and an integer $j\le q$ such that $m=p_j\le p_q$ and all the $a_i$ are distinct mod $m$.
		\end{lem}
		\begin{proof}
			Suppose that for each $1\le m=p_j\le p_q$, there exist $i\ne j$ with $a_i\equiv a_j$ (mod $m$).
			Then $\prod_{i=1}^q p_j$ divides $\prod_{i< j}(a_j-a_i)$. Since $i\ne j\implies a_i\ne a_j$ this gives,
			using \Cref{lem:MA} and the AM-GM inequality,
			\begin{eqnarray*}
				p_q\#=\prod_{i=1}^q p_j &\le& \prod_{1\le i<j\le c} a_j-a_i\\
				&\le& \left(\text{average}(a_j-a_i)\right)^{\binom{c}2}\\
				&\le&\left(\frac{c}{2(c-1)}\ell\right)^{\binom{c}2}, 
			\end{eqnarray*}
			where $\ell=a_c-a_1<n$.
			So
			\[
			\frac{2(c-1)}c (p_q\#)^{1/\binom{c}2} <n.
			\]
			Contrapositively, if all the $a_i< n\le \frac{2(c-1)}c (p_q\#)^{1/\binom{c}2}$ then
			there exists a prime $p_j\le p_q$ such that all the $a_i$ are non-congruent mod $p_j$.
		\end{proof}

		Let $\log$ with no subscript denote the natural logarithm.
		\begin{df}
			The first Chebyshev function $\vartheta:\mathbb N\to\mathbb R$ is defined by
			$\vartheta(x)=\sum_{i=1}^{q}\log p_i=\log(x\#)$ for $p_q\le x< p_{q+1}$.
		\end{df}
		\begin{thm}[{\cite[equation (3.16)]{MR137689}}]\label{rosser}
			\[
			x(1-1/\log x)<\vartheta(x),\quad \text{for all }x\ge 41.
			\]
		\end{thm}
		\begin{cor}\label{feb27-2022-2pm}
			For each $\eps>0$, and all sufficiently large $x$,
			\[
			(1-\eps)\exp(x)\le x\#.
			\]
		\end{cor}
		\begin{proof}
		\Cref{rosser} is equivalent to
			\[
			\frac{\exp(x)}{\exp(x/\log x)}<x\#,\quad \text{for }41\le x.
			\]
		\end{proof}
		\begin{cor}\label{feb27-2022-3pm}
			Let $\eps>0$ and $c\in\mathbb N$.
			Let $0\le a_1<\dots < a_c<n$, all integers, with $n$ sufficiently large.
			Let $q\in\mathbb N$ be such that
				$n\le \frac{2(c-1)}c ((1-\eps)\exp(p_q))^{1/\binom{c}2}$.
			There exists a modulus $m$ and an integer $j\le q$ such that $m=p_j\le p_q$ and all the $a_i$ are distinct mod $m$.
		\end{cor}
		\begin{proof}
		From \Cref{feb27-2022-2pm} and \Cref{feb16-2022}.
		\end{proof}

		\begin{thm}\label{feb17-pm-2022}
			Let $c\in\mathbb N$ and let $n$ be sufficiently large.
			For any $a_1<\dots<a_c<n$, for each $y\in\{0,1\}^n$, there is an $x\in\{0,1\}^n$
			which agrees with $y$ on the $a_i$, and satisfies $A_N(x)\le \binom{c}2\log n<n/3$.
		\end{thm}
		\begin{proof}
			Notice that if we take $n=\frac{2(c-1)}c ((1-\eps)\exp(p_q))^{1/\binom{c}2}$ (within $\pm 1$) in \Cref{feb27-2022-3pm} then
			solving for $p_q$, if $c\ge 2$ we have
			\begin{eqnarray*}
			p_q	&=&\log\left\{\left(\frac{c}{2(c-1)}n\right)^{\binom{c}2}\cdot \frac{1}{1-\eps}\right\}\\
				&=&\binom{c}2 \left(\log\frac{n}2 + \log\frac{c}{(c-1)}\right) - \log (1-\eps)\\
				&\le& \binom{c}2 \log n,
			\end{eqnarray*}
			and so
			\begin{equation}\label{feb17-2022}
				m\le p_q=\binom{c}2\log n.
			\end{equation}
			Take $x$ to be a (non-integral, in general) power of $z$, $x=z^{n/m}$, $\abs{z}=m$ as in~\eqref{feb17-2022}, where $z(a_i\mod m)=y(a_i)$ for each $i$.
			In other words, given $y\in\{0,1\}^n$, consider $y$ on $\{a_1,\dots,a_c\}$, copy those values $(y(a_1),\dots,y(a_c))$ to a $z\in\{0,1\}^m$, then extend that
			$z$ to $x\in\{0,1\}^n$ by repetition.
		\end{proof}
		\begin{rem}
			Let us consider an example of the construction in \Cref{feb17-pm-2022}.
			Let $c=12$, $m=14$ (admittedly not a prime, but optimal in this example),
			and let $a_i$ and $b_i=a_i$ mod $m$ for each $i$ be as follows.
			\begin{eqnarray*}
				(a_1,\dots,a_{12})	&=&	(3,4,5,7,8,11,20,23,24,26,27,28)\\
				(b_1,\dots,b_{12})	&=&	(3,4,5,7,8,11,6,9,10,12,13,0)
			\end{eqnarray*}
			We let $y$ be an interesting word from~\cite[Theorem 20]{MR3828751} and let $z$ be obtained as in \Cref{feb17-pm-2022}.
			We underline the positions $a_i$ in $x$ and $y$.
			\begin{eqnarray*}
				y &=& \mt{000\ul{1}\ul{0}\ul{1}0\ul{1}\ul{1}01\ul{0}00011001\ul{0}01\ul{1}\ul{1}1\ul{1}\ul{0}\ul{1}11}\\
				x &=& \mt{122\ul{1}\ul{0}\ul{1}0\ul{1}\ul{1}11\ul{0}10122101\ul{0}11\ul{1}\ul{1}0\ul{1}\ul{0}\ul{1}22} = z^{31/14}\\
				z &=& \mt{12210101111010}
			\end{eqnarray*}
			The occurrences of $\mt{2}$ in $z$ can be replaced by either $\mt 0$ or $\mt 1$ as they are unconstrained.
			The NFA showing that $A_N(x)\le m$ is as follows
			
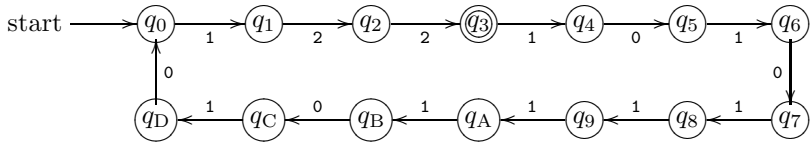
\[
				\xymatrix{
				\mathrm{start}\ar[r]&	*+[Fo]{q_0}\ar[r]_{\mt 1}	&	*+[Fo]{q_1}\ar[r]_{\mt 2}	&	*+[Fo]{q_2}\ar[r]_{\mt 2}	&	*+[Foo]{q_3}\ar[r]_{\mt 1}	&	*+[Fo]{q_4}\ar[r]_{\mt 0}	&	*+[Fo]{q_5}\ar[r]_{\mt 1}	&	*+[Fo]{q_6}\ar[d]_{\mt 0}\\
									&	*+[Fo]{q_{\mathrm D}}\ar[u]_{\mt 0}	&	*+[Fo]{q_{\mathrm C}}\ar[l]_{\mt 1}	&	*+[Fo]{q_{\mathrm B}}\ar[l]_{\mt 0}	&	*+[Fo]{q_{\mathrm A}}\ar[l]_{\mt 1}			&	*+[Fo]{q_9}\ar[l]_{\mt 1}	&	*+[Fo]{q_8}\ar[l]_{\mt 1}	&	*+[Fo]{q_7}\ar[l]_{\mt 1}
				}
			\]
			 where the final state, $q_3$, is indicated by a double circle.
		\end{rem}

		\Cref{feb17-pm-2022} showed that that it is possible for a low-complexity word to agree with a given high-complexity word in given positions.
		\Cref{thm:july31} will be used for a converse problem: showing that a high-complexity word can agree with a given low-complexity word in given positions.
		\begin{thm}[{\cite{kjos-hanssen_2021}}]\label{thm:july31}%July 31, 2019
			Let $\mathbb P_n$ denote the uniform probability measure on words $x\in\Gamma^n$, where $\Gamma$ is a finite alphabet of cardinality at least 2.
			For all $\epsilon>0$,
			\[
				\lim_{n\to\infty}\mathbb P_n\left(\left|\frac{A_N(x)}{n/2}-1\right|<\epsilon\right)=1.
			\]
		\end{thm}

	\subsection{Defining the class $\SAC^0$}

		The main fact about $\SAC^0$ that we shall use is \Cref{feb27-2022-8pm}.
		\begin{thm}\label{feb27-2022-8pm}
			If $A\subseteq\Sigma^*$ is in $\SAC^0$ then
			there is a constant $c$ such that for each $n$,
			there is a formula $\psi=\bigvee_{i=1}^k\varphi_i$ such that each $\varphi_i$ mentions at most $c$ many variables,
			and for $x=(x_1,\dots,x_n)\in\Sigma^n$, $x\in A$ iff $\psi(x)$ holds.
		\end{thm}
		\begin{rem}
			The literature is not consistent on whether $\SAC^0$ is by definition uniform, i.e., uniformly efficiently computable.
			Aaronson~\cite{Aaronson} states that ``a uniformity condition may also be imposed''.
			Our result does not require the circuit families to be uniform.
		\end{rem}
		\begin{rem}\label{mar1-2022}
			The requirement of polynomial size circuits, while presumably important for $\SAC^k$ with $k\ge 1$, is redundant for $\SAC^0$.
			When the depth is bounded by $c$, the number of variables in any $\varphi_i$ is bounded by another constant $d=d(c)$, and
			there are only $e=e(d)$ many formulas in $d$ variables, so that the $\varphi_i$ are chosen from a set of size at most $\binom{n}d e(d)$, which is a polynomial
			in $n$ of degree $d$.
		\end{rem}

		\Cref{feb27-2022-8pm} can be considered already well known.
		Borodin et al.~\cite[page 560]{MR996836}, referring back to Venkateswaran~\cite[page 383]{MR1130778} describe the semi-unbounded fan-in circuit model as follows.
		\begin{quote}
			[...] in this model, we allow OR gates with arbitrary fan-in, whereas all
			AND gates have bounded fan-in. Input variables and their negations are supplied, but
			negations are prohibited elsewhere.
		\end{quote}
		An example of a circuit computing the function $f(x_1,x_2)=x_2$, with some redundant gates:
		
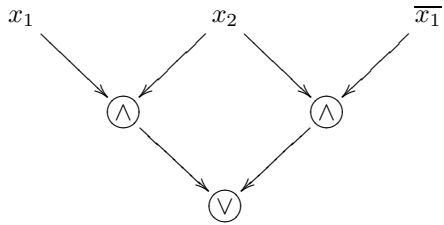
\[
		\xymatrix{
		x_1\ar[dr]	&		&x_2\ar[dl]\ar[dr]	&	&\overline{x_1}\ar[dl]\\
					& *+[Fo]{\wedge}\ar[dr]&					&*+[Fo]{\wedge}\ar[dl]&	\\
					&		&	*+[Fo]{\vee}			&\\
		}
		\]
		It is important for us to note that for the described model, all unbounded fan-in OR gates may be assumed to be at the maximum depth, i.e.,
		the OR gates is the output gate and is the furthest removed from the input variables and their negations.
		This is shown by repeated use of the distributive law for $\wedge$ and $\vee$.

	\subsection{Immunity to $\SAC^0$}

		\begin{df}
			Let $L\subseteq\Sigma^*$ and $\mathsf C\subseteq\mathcal P(\Sigma^*)$ where $\mathcal P$ denotes power set.
			Let $\mathsf{Inf}=\{L: L \text{ is an infinite set}\}$.
			\begin{itemize}
				\item The set $L$ is $\mathsf{C}$-immune if for all $C\in\mathsf{Inf}$, $C\in\mathsf{C}\implies C\not\subseteq L$.
				\item The set $L$ is $\mathsf{C}$-coimmune if $\Sigma^*\setminus L$ is $\mathsf{C}$-immune.
			\end{itemize}
			Writing $\mathsf{coC}=\{L\subseteq\Sigma^*: L\not\in\mathsf{C}\}$, we then also have:
			\begin{itemize}
				\item The set $L$ is $\mathsf{coC}$-immune if for all $C\in\mathsf{Inf}$, $C\not\in\mathsf{C}\implies C\not\subseteq L$.
				\item The set $L$ is $\mathsf{coC}$-coimmune if for all $C\in\mathsf{Inf}$, $C\not\in\mathsf{C}\implies C\not\subseteq \Sigma^*\setminus L$.
			\end{itemize}
		\end{df}
		\begin{rem}\label{feb22-2022-twosday}
			We can think of immune sets as ``small'' and coimmune sets as ``big'' in some sense.
			Conversely, if a set $L$ is both $\mathsf{C}$-immune and $\mathsf{C}$-coimmune then it points to an ability to avoid having subsets in $\mathsf{C}$,
			i.e., sets in $\mathsf{C}$ may tend to be ``big'' in some sense.

			From this point of view, \Cref{flipTurn} says that $L_{2,3}$ is neither big nor small (or if the reader prefers, both big and small),
			and many of the sets in the class $\SAC^0$ are rather large.
		\end{rem}
		\begin{thm}\label{flipTurn}
			$L_{2,3}$ is $\SAC^0$-immune and $\SAC^0$-coimmune.
		\end{thm}
		\begin{proof}
			Let $\varphi_n$ be a $\SAC^0$ formula, representing a circuit family. Then $\varphi_n=\bigvee_i \varphi_{i,n}$ where $\varphi_{i,n}$ depends only on $c$ bits of the input,
			where $c$ does not depend on $n$ but the choice of $c$ bits can depend on both $n$ and $i$.

			Suppose some $\varphi_{i,n}$ determines membership in $L_{2,3}$ based on the bits $p_1<\dots <p_{c}$ of the input.

			To show $L_{2,3}$ is $\SAC^0$-coimmune, we show that the bits $p_i$ cannot guarantee low complexity, i.e., $x\not\in L_{2,3}$.
			Let $n$ be sufficiently large.
			Fix values $x(p_i)$ and select the values $x(k)$, $k\not\in\{p_1,\dots,p_c\}$, $1\le k\le n$, randomly (uniformly and independently).
			The probability of the values $x(p_i)$ is $2^{-c}$.
			Let $\epsilon=1/3$.
			Then
			\[
			\left|\frac{A_N(x)}{n/2}-1\right|<\epsilon \implies A_N(x)\ge  \frac{\lfloor n/2+1\rfloor}2>n/3 \implies x\in L_{2,3},
			\]
			By \Cref{thm:july31}, for large enough $n$,
			\[
				\mathbb P_n\left(\left|\frac{A_N(x)}{n/2}-1\right|<\epsilon\right)\ge 1-2^{-(c+1)}.
			\]
			Therefore, by the union bound, the probability of either disagreeing with some $x(p_i)$, or having low complexity, is at most
			\[
				(1-2^{-c}) + 2^{-(c+1)} < 1,
			\]
			so the probability of agreeing with the $x(p_i)$ and also having high complexity is positive.
			
			Since this occurs with positive probability, in particular it occurs for at least one $x$. For that $x$, $x\in L_{2,3}$, as desired.
			This shows that no single $\varphi_{i_0,n}$ can imply that $x\not\in L_{2,3}$.
			Since $\varphi_{i_0,n}$ implies $\bigvee_i\varphi_{i,n}$, neither can $\bigvee_i\varphi_{i,n}$ imply that $x\not\in L_{2,3}$.
			So $L_{2,3}$ is $\coSAC_0$-immune.

			To show $L_{2,3}$ is $\SAC^0$-immune, we show that the bits $p_i$ cannot guarantee high complexity, i.e., $x\in L_{2,3}$.

			This follows from \Cref{feb17-pm-2022}, which shows how in the limit we can force $x\not\in L_{2,3}$.
		\end{proof}

	\section{The class $\OSAC^0$}

		G\'al and Wigderson~\cite{MR1611744} considered arithmetic circuits with gates from the basis $\{+,-,\times\}$ over fields such as $\mathbb Z/2\mathbb Z$,
		i.e., $GF(2)$.
		All constants of the field may be used.
		Boolean circuits have the standard Boolean basis $\{\wedge,\vee,\neg\}$.
		\emph{Semi-unbounded fan-in} circuits have
		constant fan-in $\times$ (resp.~$\wedge$) gates and
		unbounded fan-in $+$ (resp.~$\vee$) gates.
		Semi-unbounded fan-in Boolean circuits may have negations only at the input level.
		$\SAC^k$ denotes the class of languages accepted by polynomial size, depth $O((\log n)^k)$ semi-unbounded fan-in Boolean circuits.
		$\OSAC^k$ denotes the class of languages accepted by polynomial size, depth $O((\log n)^k)$ semi-unbounded fan-in arithmetic circuits over $GF(2)$.

		G\'al and Wigderson showed that $\SAC^1\subseteq\OSAC^1$.
		We note that $\SAC^0\ne\coSAC^0$ was shown in~\cite{MR996836}.
		However, $\mathrm{co}\OSAC^0 = \OSAC^0$ since if $\varphi$ is an multilinear polynomial representing $L$ then $\varphi+1$ represents $\Sigma^*\setminus L$.
		%Let $\oplus$, also $+$, denote XOR, the exclusive disjunction.
		Analogously to \Cref{feb27-2022-8pm} for $\SAC^0$, the main fact about $\OSAC^0$ that we shall use is \Cref{feb27-2022-later}.
		\begin{thm}\label{feb27-2022-later}
			If $A\subseteq\Sigma^*$ is in $\OSAC^0$ then there is a constant $c$ such that for each $n$,
			there is a formula $\psi=\bigoplus_{i=1}^k\varphi_i$ such that each $\varphi$ mentions at most $c$ many variables,
			and for $x=(x_1,\dots,x_n)\in\Sigma^n$, $x\in A$ iff $\psi(x)$ holds.
		\end{thm}
		\begin{proof}
			A formula of bounded depth with bounded fan-in can only mention a bounded number of variables.
			All the unbounded fan-in occurrences of $\bigoplus$ can be folded into one using the distributive law of $\cdot$ and $+$.
		\end{proof}

		\begin{rem}
			In \Cref{feb27-2022-later} it follows that the formulas $\psi$ have size polynomial in $n$, much as in \Cref{mar1-2022}.
		\end{rem}

		We can show that $\{x:A_N(x)>1\}\in \SAC^0\setminus\OSAC^0$ and hence $\SAC^0\not\subseteq\OSAC^0$.
		\Cref{feb24-2022} resolves an open implication from the Complexity Zoo~\cite{Aaronson}.

		\begin{df}
			We define the \emph{degree} of a multilinear polynomial by
			\[
				\deg\left(\sum_{F\in\mathcal F}\prod_{i\in F} x_i\right)=\max\{\abs{F}:F\in\mathcal F\}.
			\]
		\end{df}
		\begin{lem}\label{feb26-2022-11am}
			A multilinear polynomial over the ring $(\mathbb Z/2\mathbb Z,+\cdot)$ is identically 0 as a function over $\mathbb Z/2\mathbb Z$ only if all the coefficients are 0.
		\end{lem}
		\begin{proof}
			Let $P$ be the set of (formal) multilinear polynomials in $n$ variables with coefficients in $\{0,1\}$.
			Let $F_P$ be the set of Boolean functions computed by elements of $F$.
			It suffices to show that two multilinear polynomials are equal as functions only if they are equal as (formal) polynomials,
			i.e., the map sending a formal polynomial to its function is one-to-one.
			Thus, it suffices to show that $\abs{F_P}\le\abs{P}$.

			Let $B$ be the set of all Boolean functions in $n$ variables.
			Each Boolean function may be expressed as a multilinear monomial over $\mathbb Z/2\mathbb Z$.
			To wit,
			\[
				\bot = 0,\quad\top = 1,\quad\neg a = a + 1,\quad\text{and}\quad a\wedge b = a \cdot b.
			\]
			Thus $B=F_P$.

			For each set $S\subseteq\{1,\dots,n\}$ there is a multilinear monomial $\prod_{i\in S}x_i$.
			Thus there are $2^n$ multilinear monomials in $n$ variables. Any subset $\mathcal S$ of these may be included in a multilinear polynomial
			\[
			\sum_{S\in\mathcal S}\prod_{i\in S}x_i.
			\]
			Thus $\abs{P}=2^{2^n}$.
			As it is well known that $\abs{B}=2^{2^n}$, we conclude
			\[
				\abs{F_P} = \abs{B} = 2^{2^n} = \abs{P}.\qedhere
			\]
		\end{proof}
		\begin{rem}
			\Cref{feb26-2022-11am} is a known result, but it should not be confused with the similar representation of Boolean functions as multilinear polynomials over $\mathbb R$.
			For instance, XOR is $x+y-xy$ over $\mathbb R$, but $x+y$ over $GF(2)$.
			Also, \Cref{feb26-2022-11am}  is somewhat sharp in that if we go beyond linear polynomials we quickly get a counterexample: $x^2+x$ is identically 0 as a function.
		\end{rem}
		\begin{thm}\label{feb24-2022}
			$\SAC^0\not\subseteq\OSAC^0$.
		\end{thm}
		\begin{proof}
			We will show that the family of disjunction functions
			\[
				\bigvee_n(x_1,\dots,x_n)=x_1\vee\dots\vee x_n
			\]
			is in $\SAC^0\setminus\OSAC^0$.
			To see that $\bigvee_n$ is in $\SAC^0$ is trivial: it is computed by the following constant-depth circuit.
			
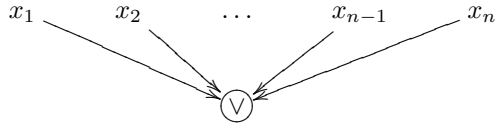
\[
				\xymatrix{
					x_1\ar[drr] & x_2\ar[dr]	& \dots & x_{n-1}\ar[dl]&x_n\ar[dll]\\
						&		& *+[Fo]{\vee}	&
				}
			\]
			It remains to show that the function
			$\bigvee_n$ is not in $\OSAC^0$. By de Morgan's law,
			\begin{equation}\label{eq:office}
				\bigvee_n (x_1,\dots,x_n)=1+\prod_{i=1}^n(1+x_i)=1+\sum_{F\subseteq [n]}x_F = \sum_{\emptyset\ne F\subseteq [n]}x_F
			\end{equation}
			where $x_F=\prod_{i\in F}x_i$, $x_{\emptyset}=1$. For instance, $\bigvee_2(x,y)=1+(xy+x+y+1)=xy+x+y$.

			By \Cref{feb26-2022-11am}, if $\bigvee_n$ is equal to the function expressed by a $\OSAC^0$ formula (circuit) then that formula expands to the polynomial in~\eqref{eq:office}.
			By \Cref{feb27-2022-later} an $\OSAC^0$ formula is a sum of terms $\sum_i\varphi_i$ (mod 2)
			where each $\varphi_i$ depends on a bounded number of variables.
			So when we express $\varphi_i$ as a multilinear polynomial, it will have bounded degree
			since
			\[
			\deg\left(\sum_i\varphi_i\right)\le\max_i\deg(\varphi_i).
			\]
			Since $\bigvee_n$ has degree $n$ by~\eqref{eq:office}, the result follows.
		\end{proof}
		\begin{cor}\label{mar1-20222pm}
			$\{x\in\{\mt 0,\mt 1\}^*: A_N(x)=1\}\not\in\OSAC^0$.
		\end{cor}
		\begin{proof}
			We have $A_N(x)=1$ iff $x\in\{\mt 0^n,\mt 1^n\}$ (where $n=\abs{x}$), iff
			\[
			\prod_{i\in [n]}x_i + \prod_{i\in [n]} (x_i+1) = \sum_{\emptyset\ne F\subsetneq [n]} \prod_{i\in F}x_i = 1.
			\]
			Since $\sum_{\emptyset\ne F\subsetneq [n]} \prod_{i\in F}x_i$ has degree $n-1$, which is unbounded as $n\to\infty$, the result follows as in \Cref{feb24-2022}.
		\end{proof}
		While we have not settled whether $L_{2,3}\in\OSAC^0$, \Cref{mar1-20222pm} does show that $\OSAC^0$ is too limited to capture the computational complexity of $A_N$.
	\bibliographystyle{plain}
	\bibliography{CFL}
\end{document}